\documentclass[12pt]{iopart}
\usepackage{iopams}  
\usepackage{graphicx}
\usepackage{amsthm,amssymb}
\usepackage{latexsym}
\usepackage{bbold}
\usepackage{bbm}
\usepackage{setstack}

\newcommand{\cc}[1]{\mathcal{#1}}

\newcommand{\ket}[1]{| #1 \rangle}
\newcommand{\bra}[1]{\langle #1 |}

\def\>{\rangle}
\def\<{\langle}
\newcommand{\proj}[1]{| #1 \rangle   \langle #1 |}
\renewcommand{\rho}{\varrho}

\def\tr{{\rm Tr}}
\def\pf{{\rm Pf}}
\def\ii{{\rm i}}
\def\conv{{\rm conv}}
\def\textbf#1{{\bf #1}}
\newcommand{\Cx}{\mathbbm{C}}

\newcommand{\Rl}{\mathbbm{R}}

\newcommand{\bb}[1]{\mathbf{#1}}

\newtheorem{lemma}{Lemma}
\newtheorem{theorem}{Theorem}
\newtheorem{proposition}{Proposition}
\newtheorem{definition}{Definition}

\newtheorem{corollary}{Corollary}
\newtheorem{program}{Program}
\newcounter{rem}
\newenvironment{remark}{ \emph{Remark \arabic{rem}:}}{\medskip\addtocounter{rem}{1}}

\def\beq{\begin{equation}}
\def\eeq{\end{equation}}
\def\be{\begin{equation}}
\def\ee{\end{equation}}
\def\ben{\begin{eqnarray}}
\def\een{\end{eqnarray}}
\def\beqa{\begin{eqnarray}}
\def\eeqa{\end{eqnarray}}
\def\eea{\end{array}}
\def\bea{\begin{array}}
\newcommand{\bei}{\begin{itemize}}
\newcommand{\eei}{\end{itemize}}
\newcommand{\bee}{\begin{enumerate}}
\newcommand{\eee}{\end{enumerate}}
\def\bep{\begin{proposition}}
\def\eep{\end{proposition}}
\def\bel{\begin{lemma}}
\def\eel{\end{lemma}}
\def\bet{\begin{theorem}}
\def\eet{\end{theorem}}
\def\bed{\begin{definition}}
\def\eed{\end{definition}}

\begin{document}

\title[The Power of Noisy Fermionic Quantum Computation]{The Power of Noisy Fermionic Quantum Computation}
\author{Fernando de Melo$^1$, Piotr \'{C}wikli\'{n}ski$^2$, Barbara M. Terhal$^2$}
\address{$^1$Centrum Wiskunde \& Informatica,\\ Science Park 123, 1098 XG Amsterdam, The Netherlands\\
$^2$Institute for Quantum Information, RWTH Aachen University,\\ 52056 Aachen, Germany}
\eads{\mailto{melo@cwi.nl}, \mailto{cwiklinski@physik.rwth-aachen.de}, \mailto{terhal@physik.rwth-aachen.de}}

\date{\today}
\begin{abstract}
We consider the realization of universal quantum computation through braiding of Majorana fermions supplemented by unprotected preparation of noisy ancillae. It has been shown by Bravyi~[\emph{Phys.~Rev.~A} \textbf{73}, 042313 (2006)] that under the assumption of perfect braiding operations, universal quantum computation is possible if the noise rate on a particular 4-fermion ancilla is below $40\%$. We show that beyond a noise rate of $89\%$ on this ancilla the quantum computation can be efficiently simulated classically: we explicitly show that the noisy ancilla is a convex mixture of Gaussian fermionic states in this region, while for noise rates below $53\%$ we prove that the state is not a mixture of Gaussian states. These results are obtained by generalizing concepts in entanglement theory to the setting of Gaussian states and their convex mixtures. In particular we develop a  set of criteria, namely the existence of a Gaussian-symmetric extension, which determine whether a state is a convex mixture of Gaussian states.
\end{abstract}



\section{Introduction}

One interesting route towards universal quantum computation is through the realization of Majorana fermion qubits. Such Majorana fermion qubits, encoded in pairs of nonlocal fermionic zero-energy modes, are believed to be present in various quantum systems such as a $\nu=5/2$ fractional Quantum Hall system, $p_x+\ii p_y$ superconductors, and recently proposed topological insulator/superconductor and semiconducting nanowires/superconductor structures (see~\cite{alicea:review} for a review). Braiding operations on such Majorana fermion qubits can implement certain topologically-protected gates, namely single-qubit Clifford gates~\cite{bravyi:uni_maj,Ahlbrecht09}. A system of Majorana fermions initially prepared in a fermionic Gaussian state (see definition below) and undergoing braiding operations (which are a subset of non-interacting fermion operations~\footnote{Non-interacting fermionic operations are unitary transformations generated by quadratic Hamiltonians on the creation and annihilation operators of  fermionic modes. Such Hamiltonians can always be block-diagonalized, and in this new basis the different fermionic modes do not interact.}) can be classically efficiently simulated if it is not supplemented by additional resources. Universal quantum computation can be achieved if this supplement consists, for example, of either (i) gates which use quartic interaction between Majorana fermions~\cite{Bravyi2002210} or (ii) a quartic parity measurement~\cite{Beenakker04} or (iii) two ancillae in certain special states $\ket{a_4}$ and $\ket{a_8}$, involving respectively 4 and 8 Majorana fermions~\cite{bravyi:uni_maj}. The advantage of the last realization is that even when these ancillae are noisy, one can purify them using braiding operations. This distillation of clean ancillae states, which are then used to implement the required gates for universality, is similar to the magic-state-distillation scheme by which Clifford group operations are used to distill almost noise-free single qubit $\pi/8$ ancillae from noisy ones~\cite{BK:magicdistill}. For the model of computation based on braiding of Majorana fermions, distilled  $\ket{a_4}$ states give the ability to perform  single quibt $\pi/8$ gates, while distilled  $\ket{a_8}$ states allow for the implementation of  $CNOT$ gates~\cite{bravyi:uni_maj}. The noise threshold for such distillation schemes, which assume that the distilling gates and operations are noise-free, is of the order of tens of percents which makes them attractive. But one can also ask the converse question: how noisy are these ancillae allowed to be before one can efficiently classically simulate the entire quantum computation? This question has been addressed in the case of noisy single-qubit magic states and noisy single-qubit gates \cite{reichardt:distill, buhrman+:FT,DK:tight}. 

In this paper we consider a similar question for computation using Majorana fermions. It is not hard to show (see Section~\ref{sec:discussion}) that if the noisy ancillae are convex mixtures of Gaussian fermionic states and the computation involves only non-interacting fermionic operations, one can still classically efficiently simulate such quantum computation. Hence we set out to develop a criterion that determines whether a state is a convex mixture of Gaussian fermionic states in Section~\ref{sec:GaussDeco}. We find such a criterion in the form of a hierarchy of semidefinite programs, similar as for separable states \cite{DPS:prl}. Unfortunately, the computational effort for implementing this general criterion is too large to give decisive information for our problem of interest and we have recourse to an analytical approach for the noisy $\ket{a_8}$ ancilla in Section~\ref{sec:appl}. Even though our general criterion is not immediately useful for the problem at hand, its generality and similarity with separability criteria makes it interesting in itself.

Our work is different from previous work on criteria which determine whether a fermionic state with a fixed number of fermions has a single Slater determinant, or whether a fermionic state can be written as a convex combination of states with single Slater determinant \cite{SLM:slater_det, eckert+:indis}. The important distinction is that we fix only the parity of the fermionic state and not the number of fermions. Our goal is to extend the class of Gaussian fermionic states in a natural way, by considering states which can be written as convex combinations of Gaussian states. We  call such states convex-Gaussian or `having a Gaussian decomposition'. Even though physical states have a fixed number of fermions, Gaussian fermionic states are important approximations to physically non-trivial states, such as the superconducting BCS state. Our criterion thus intends to separate Gaussian states and convex combinations thereof, from fermionic states with a richer structure which cannot be simply viewed as states in which fermions are paired \footnote{Any Gaussian fermionic state with an even number of fermions can be brought, by fermion-number preserving quadratic interactions, to a normal form which is a superposition of states with fixed pairs of fermions created from the vacuum, see e.g. \cite{kraus:thesis}.}. Note that the question of having a Gaussian decomposition is also different from the question of {\em pairing} which is analyzed in \cite{kraus:pairing}: Gaussian fermionic states can be paired while single Slater determinant states cannot. We hope that our results in separating convex-Gaussian states from fermionic states with a richer structure may lead to tools for understanding ground states of interacting fermion systems beyond mean-field theory.


\section{Preliminaries: Fermionic Gaussian states}
\label{sec:FermionGauss}

We consider a system of $m$ fermionic modes, with corresponding creation ($a^\dagger_k$) and annihilation ($a_k$) operators ($k=1,\ldots,m$), respecting the Fermi-Dirac anti-commutation rules,  $\{a_j,a_k\}=0$ and $\{a_j, a^\dagger_k\}=\delta_{jk} I$. For systems in which fermionic parity is conserved, it is more convenient to use the $2m$ Majorana fermion operators defined as $c_{2k-1}=a_k+a_k^{\dagger}$ and $c_{2k}=-\ii(a_k^\dagger-a_k)$. The operators $\{c_i\}_{i=1}^{2m}$ are Hermitian, traceless and form a Clifford algebra $\cc{C}_{2m}$ with $\{c_j,c_k\}=2 \delta_{jk}I$~\footnote{A system of $m$ qubits is isomorphic with a system with $m$ fermions and the unitary Jordan-Wigner transformation maps each Majorana fermion operator onto a nonlocal product of Pauli operators acting on $m$ qubits.}.

Any Hermitian operator $X\in \cc C_{2m}$  that can be written as the linear combination of products of an even number of Majorana operators is called an \emph{even} operator, i.e. 
\begin{equation}
X=\alpha_0 I +\sum_{k=1}^{m} \ii^k \sum_{1 \leq a_1 < a_2 < \ldots < a_{2k} \leq 2m} \alpha_{a_1,a_2,\ldots,a_{2k}} c_{a_1} c_{a_2} \ldots c_{a_{2k}},
\label{anyX}
\end{equation}
where the coefficients $\alpha_0$ and $\alpha_{a_1,a_2,\ldots,a_{2k}}$ are real, for $(c_{a_1} c_{a_2} \ldots c_{a_{2k}})^{\dagger}=(-1)^k c_{a_1} c_{a_2} \ldots c_{a_{2k}}$. The parity of the number of fermions is conserved by the action of an even operator $X$ as $X$ commutes with the fermionic number-parity operator $C_{\rm all}=\ii^{m} c_1 c_2 \ldots c_{2m} = \prod_{i=1}^m(I -2 a_i^\dagger a_i)$. Thus the projector onto a pure state $\ket{\psi}$ with {\em fixed} parity is an even fermionic operator, while $\ket{\psi}$ has eigenvalues $C_{\rm all}=\pm 1$ depending on whether the parity of the number of fermions in $\ket{\psi}$ is odd or even.

Given an even state $\rho \in \cc C_{2m}$, the correlation matrix $M$ is a $2 m \times 2m$ real, anti-symmetric matrix with elements 
\beq
M_{ab} := \frac{\ii}{2} \tr (\rho[c_a, c_b]),\;\; \text{with }\;\; a,b=1,\ldots, 2m.
\label{def:M}
\eeq
Real, anti-symmetric matrices such as $M$ can be brought to block-diagonal form by a real orthogonal transformation $R \in SO(2m)$, i.e. 
\begin{equation}
M= R \;\bigoplus_{j=1}^m \left(\begin{array}{cc} 0 & \lambda_j \\  -\lambda_j & 0 \end{array} \right) R^T.
\label{blockdiag}
\end{equation}

Let us now define fermionic Gaussian states. Fermionic Gaussian states $\rho$ are even states of the form $\rho=K \exp(-\ii \sum_{i \neq j} \beta_{ij} c_i c_j)$ with real anti-symmetric matrix $\beta_{ij}$ and normalization $K$. Hence Gaussian fermionic states are ground-states and thermal states of non-interacting fermion systems. One can block-diagonalize $\beta$ and re-express $\rho$ in {\em standard form} as 
\begin{equation}
\rho=\frac{1}{2^m} \prod_{k=1}^m \chi_k, \;\;{\rm where }\;\; \chi_k=I+\ii \lambda_k \tilde{c}_{2k-1} \tilde{c}_{2k},\,
\label{sform}
\end{equation}
where $\tilde{c}=R^T c$ with $R$ block-diagonalizing the matrix $\beta_{ij}$. The coefficients $\lambda_j$ (which can be expressed in terms of the eigenvalues of $\beta_{ij}$) will lie in the interval $[-1,1]$. For Gaussian pure states, $\lambda_j \in \{-1,1\}$ so that $M^T M=I$, while for Gaussian mixed states $M^T M < I$. In the theory of Gaussian fermionic states, a special role is played by those unitary transformations which map Gaussian states onto Gaussian states. These are the transformations generated by Hamiltonians of non-interacting fermionic systems, i.e. Hamiltonians which are quadratic in Majorana fermion operators. In this paper we will refer to these unitary transformations as fermionic linear optics (FLO) transformations, as they, similarly as for bosonic linear optics transformations,  have the property that 
\begin{equation}
U c_i U^{\dagger}=\sum_{j} R_{ij} c_j,
\label{eq:FLOrule}
\end{equation}
where $U$ is a FLO transformation and $R \in SO(2m)$. Hence the transformations which block-diagonalize $M$ (or $\beta$) are FLO transformations. Note that $C_{\rm all}$ is invariant under any unitary FLO transformation $U$, as $U C_{\rm all}=C_{\rm all} U$.

It is clear from the standard form, Eq.(\ref{sform}), that a fermionic state is fully determined by its correlation matrix $M$. The expectation of higher order correlators for a Gaussian fermionic state are efficiently obtained by virtue of Wick's theorem
\beq
 \tr(\rho\, c_{a_1} c_{a_2}\ldots c_{a_{2k}}) = \ii^{-k} \pf(M|_{a_1,\ldots, a_{2k}}),
\label{wick}
\eeq
where $M|_{a_1,\ldots, a_{2p}}$ is the sub-matrix of $M$ which contains only the elements $M_{jk}$ with $j,k=a_1,\ldots, a_{2p}$ and $\pf(.)$ is the Pfaffian~\footnote{Let $A=(a_{i,j})$ be a $2m\times 2m$ anti-symmetric matrix, then $\pf(A) = \frac{1}{2^m m!}\sum_{\pi\in  S_{2m}} {\rm sgn}(\pi)\prod_{i=1}^{m} a_{\pi(2i-1),\pi(2i)}$, where $S_{2m}$ is the set of all permutations of $2m$ symbols. The Pfaffian is non-zero only for anti-symmetric matrix of even dimension. }. 

Their efficient description~(\ref{def:M}) combined with the possibility to efficiently evaluate the expectation value of observables~(\ref{wick}) makes fermionic Gaussian states a valuable  tool for {\em approximating} ground-states, thermal states or dynamically-generated states using (generalized) Hartree-Fock methods of interacting fermion systems, see e.g. \cite{bach94,kraus10,kraus:thesis, kittel:qt_solids}. Their concise representation together with Eq.(\ref{eq:FLOrule}) also allows for an efficient classical simulation of quantum computations that employ only FLO operations and are initialized with Gaussian states~\cite{BarbaraFLO, bravyi:FLO,JozsaMiyake, BravyiKoenig}.


\section{Beyond Gaussian states:  Gaussian decompositions}
\label{sec:GaussDeco}

Given all the applications of Gaussian states, it is natural to try to enlarge this set of states to form a convex set, in analogy with extending product states to separable states. A first observation in this direction is that mixed Gaussian states are straightforward convex mixtures of pure Gaussian states.  This is easily seen using the standard form~(\ref{sform}), i.e.,  any Gaussian state can be written as  $\rho=\prod_{k=1}^m \chi_k/2^m$ with $\chi_k=I+\ii (2p_k-1) c_{2k-1}c_{2k}=p_k (I+\ii c_{2k-1}c_{2k})+(1-p_k)(I-\ii c_{2k-1} c_{2k})$ where $0\le p_k\le 1$ (by setting $p_k=1/2$ we obtain the state $I/2^m$). However, the convex mixture of two Gaussian states is  in general not a Gaussian state~\cite{wolfe75,MarkBook}. This motivates the following definition

\begin{definition}[Convex-Gaussian]
An even density matrix  $\rho\in\cc C_{2m}$ is convex-Gaussian iff it can be written as a convex combination of pure Gaussian states in $\cc C_{2m}$, i.e. $\rho=\sum_i p_i \sigma_i$ where $\sigma_i$ are pure Gaussian states, $p_i \geq 0$ and $\sum_i p_i=1$.
\end{definition}

\begin{remark} Since every even state in $\cc C_{2m}$ is defined by $2^{2m}-1$ real coefficients, as is clear from Eq.~(\ref{anyX}) fixing the normalization $\alpha_0=1/2^m$, Carath\'{e}odory's theorem guarantees us that for any convex-Gaussian state there exists a Gaussian decomposition with at most $2^{2m}$ pure Gaussian states.
\end{remark}

\begin{remark} There is both similarity and difference between the set of separable states and the set of convex-Gaussian states. For example, for separable states local unitary transformations captured by $O(m)$ parameters relate the extreme points of the set, whereas for convex-Gaussian states the extreme points are the pure Gaussian states which are related by FLO transformations, captured by $O(m^2)$ parameters. Note that the pure states $\sigma_i$ in decomposition of a convex-Gaussian state will typically have a correlation matrix $M(\sigma_i)$ which is block-diagonal in a different $\{c_i\}$ basis. If we were to restrict ourselves to convex combinations of Gaussian states which are in standard form using one and the same set $\{c_i\}$, we would obtain a convex set isomorphic to the set of separable states diagonal in the classical bit-string basis.
\end{remark}

We will now prove that for systems of $1,2$ or $3$ fermionic modes, all  pure even states are Gaussian (a different proof of this result can be found in~\cite{bravyi:fermchannels}). From this observation it is then immediate that any even state $\rho \in \cc C_{2m}$, with $m=1,2,3$, is convex-Gaussian.   In order to prove this result we start with a useful known Lemma, reproduced here for completeness, which shows that any fermionic even density matrix has a correlation matrix with eigenvalues in the complex $[-\ii,\ii]$ interval. 

\begin{lemma}
The antisymmetric correlation matrix $M$ of any even density matrix $\rho\in \cc C_{2m}$ has eigenvalues $\pm \ii \lambda_k$ with $\lambda_k \in [-1,1]$ and $k\in \{1,...,m\}$.  We have $M^T M \leq I$ with equality if and only if $\rho$ is a Gaussian pure state. 
\label{lem:M}
\end{lemma}

\begin{proof}
We present two proofs. The correlation matrix of an even density matrix is an antisymmetric matrix, and hence it can be block-diagonalized as in Eq.(\ref{blockdiag}). Its eigenvalues are thus $\pm \ii \lambda_k$ and $\lambda_k \in [-1,1]$, as from Eq.(\ref{def:M}) $\lambda_k$ is the expectation of the Hermitian operator $\ii \tilde{c}_{2k-1} \tilde{c}_{2k}$ which has $\pm 1$ eigenvalues. We can also prove the Lemma by using a dephasing approach introduced in~\cite{WW:clifford}. We provide this alternative proof as it clearly shows that $M^T M=I$ iff $\rho$ is a pure Gaussian state, and  because we later use elements of this proof in Proposition \ref{prop:123Gauss}. Let $\rho\in \cc C_{2m}$ be any density matrix and let $\{c_i\}_{i=1}^{2m}$ be the set of Majorana fermions in which $M$ is block-diagonal (if the choice for $\{c_i\}$ is not unique, pick one). Define the FLO transformation $U_k$, $k=1,\ldots,m$ as 
\begin{eqnarray}
U_k c_{2k} U_k^{\dagger}=-c_{2k},\; U_k c_{2k-1} U_k^{\dagger}=-c_{2k-1}, \; U_k c_i U_k^{\dagger}=c_i\;\; \forall i\neq 2 k, 2k-1.
\end{eqnarray} 
Note that the FLO transformation  $U_k$ induces an orthogonal transformation $R_k$ that leaves the  correlation matrix $M$ of $\rho$ invariant. Let $\rho_k =\frac{1}{2}(\rho_{k-1} +U_k \rho_{k-1} U^\dagger_k)$, with $\rho_0 = \rho$. The iteration in $k$  has the effect of dephasing $\rho$ in the eigenbasis of each $\ii c_{2k-1} c_{2k}$. Thus $\rho_m$ is a density matrix which contains only the mutually commuting operators $\ii c_{2k-1}c_{2k}$, and hence its eigendecomposition involves the eigenstates of these operators which are Gaussian pure states. Note that $\rho_m$ is not necessarily a Gaussian state, but its eigenstates are Gaussian pure states, each of which has correlation matrix block-diagonal in the same basis. This implies that the correlation matrix $M$ is that of a Gaussian mixed state.  Note that the dephasing procedure can only increase the entropy of $\rho$. Therefore, if $M^T M=I$ at the end of this procedure, then the state at the end of the procedure, $\rho_m$, must be a pure Gaussian state. By the same token, $\rho$ at the beginning of the procedure (with the same $M^T M=I$) must be pure Gaussian. Thus $M^T M=I$ iff $\rho$ is a pure Gaussian state.
\end{proof}

\begin{proposition}
Any even pure state $\proj{\psi}\in \cc C_{2m}$ for $m=1, 2, 3$ is Gaussian.
\label{prop:123Gauss}
\end{proposition}

\begin{proof}
The statement is trivial for $m=1$. Consider $m=2$ and observe that when we block-diagonalize its correlation matrix, any even pure state can written as some linear combination of $I$, quadratic terms $\ii c_{2k-1}c_{2k}$ and $C_{\rm all}$. The dephasing procedure in the proof of Lemma \ref{lem:M} leaves any such even density 
matrix invariant -- remember that the operator $C_{\rm all}$ is invariant under FLO transformations. As the output is always convex-Gaussian,  the input state was a pure Gaussian state. The proof for $m=3$ follows the same structure as for $m=2$, but employs the additional observation that $C_{\rm all}\proj{\psi}=\pm\proj{\psi}$ for even pure states. Any pure state density matrix  $\ket{\psi}\bra{\psi}\in \cc C_6$, when written in the basis in which its correlation matrix is block-diagonal, can be expressed as  $\ket{\psi}\bra{\psi}=\alpha I+\beta C_{\rm all}+ \ii\sum_k \gamma_k c_{2k-1} c_{2k} +\sum_{i < j < k < l} \eta_{ijkl} c_i c_j c_k c_l$. The condition that $C_{\rm all} \ket{\psi}\bra{\psi}= \pm \ket{\psi}\bra{\psi}$ implies that the quartic terms only have paired correlators $c_{2k-1} c_{2k}$: this is because $C_{\rm all}$ interchanges the quartic and quadratic terms. Thus $\ket{\psi}\bra{\psi}$ has  quartic terms which only involve commuting operators 
$c_{2k-1}c_{2k}$. As argued in the proof of Lemma~\ref{lem:M},  $\ket{\psi}\bra{\psi}$ is then a convex mixture of pure Gaussian states, and the result follows.
\end{proof}

As for separable states~\cite{Karol98,braunstein+:ball,GB:Ball}, the volume of convex-Gaussian states is non-zero for any $m$: any state $\rho \in \cc C_{2m}$ is convex-Gaussian if it is ``close enough'' to the maximally mixed state $I/2^m$, which is Gaussian. The following Lemma proves this, using arguments similar as those in Ref.~\cite{braunstein+:ball}:
\begin{theorem} For every even state $\rho \in \cc C_{2m}$ there exists a sufficiently small $\epsilon > 0$ such that $\rho_{\epsilon}=\epsilon \rho +(1-\epsilon) I/2^m$ is convex-Gaussian.
\label{lem:ball}
\end{theorem}

\begin{proof}
Consider the set $\{\sigma(\vec{\lambda},\pi)\}$ of pure Gaussian states, with each $\sigma(\vec{\lambda},\pi) \in \cc C_{2m}$ defined as
\beq
\sigma(\vec{\lambda}, \pi) = \frac{1}{2^m} \prod_{k=1}^m (I +\ii \lambda_k c_{\pi(2k-1)}c_{\pi(2k)}),
\label{eq:defsig}
\eeq
where $\vec{\lambda} \in \{-1,1\}^m$ and $\pi\in S_{2m}^+$, with $S_{2m}^+$ the set of permutations with positive signature (these correspond to  $SO(2m)$ rotations of the correlation matrix, i.e.,  FLO operations). There are $(2m)!/m!$ pure states in this set and we will show that they form an over-complete basis for $\cc C_{2m}$. We prove this by showing that {\em any} Hermitian $\rho$ can be written as
\begin{equation}
\rho=\sum_{(\vec{\lambda},pi)} w_{(\vec{\lambda},\pi)} \sigma(\vec{\lambda},\pi)-c \frac{I}{2^m},
\label{decomp}
\end{equation}
where $w_{(\vec{\lambda},\pi)}\geq 0$ and $c \geq 0$. Furthermore, we can write 
\begin{equation*}
\frac{I}{2^m} = \frac{m!}{(2m)!} \sum_{\vec{\lambda},\pi} \sigma(\vec{\lambda}, \pi).
\end{equation*}
If we find the decomposition in Eq.~(\ref{decomp}) with a certain $c$ (which one could try to minimize in our construction of the decomposition), it then follows immediately that $\frac{1}{1+c}\rho+\frac{c}{1+c} \frac{I}{2^m}$ is convex-Gaussian. Taking $\epsilon=\frac{1}{1+c}$ we obtain our result. One could in principle upper bound $c$ as a function of $m$. Eq.~(\ref{decomp}) can be proved by considering how to express a single correlator $\alpha_{a_1,\ldots, a_{2k}} c_{a_1} c_{a_2} \ldots c_{a_{2k}}$ in terms of Gaussian states and $I$. Let us call the subset $\{a_1, \ldots, a_{2k}\}=S$ and the coefficient $\alpha_{a_1,\ldots,a_{2k}}=\alpha$. Let $C_S=c_{a_1} c_{a_2} \ldots c_{a_{2k}}$. We will use Gaussian mixed states $\xi(\vec{\lambda}, \pi) = \frac{1}{2^m} \prod_{k=1}^m (I +\ii \lambda_k c_{\pi(2k-1)}c_{\pi(2k)}) $ where $\lambda_i \in [-1,1]$; such states can in turn be expressed as convex mixtures of pure states $\sigma(\vec{\lambda},\pi)$ for $\lambda_i \in\{-1,1\}$ as in Eq.~(\ref{eq:defsig}). We have 
\begin{equation}
\ii^k {\rm Tr} \,C_{S} \xi(\vec{\lambda}, \pi)=\Pi_{k:\pi(2k-1), \pi(2k) \in S}\lambda_k.
\label{corrcontrib}
\end{equation}
By choosing $\lambda_k\in \{-1,1,0\}$ we can make sure that (i) $|\alpha| \xi(\vec{\lambda}, \pi)=w_{\vec{\lambda},\pi} \xi(\vec{\lambda},\pi)$ has the correct expectation $\alpha$ for the correlator $C_S$ (note that the sign of $\lambda_k$ can be chosen to fit the sign of $\alpha$) and (ii) $\xi(\vec{\lambda},\pi)$ has {\em no higher-order correlators} (by choosing some $\lambda_k=0$). The state $\xi(\vec{\lambda},\pi)$ will then also have {\em lower-order} non-zero correlators $C_{S'}$ for $S' \subseteq S$, including $S'=\emptyset$ corresponding to $I$. Thus we repeat the procedure with these various lower-order correlators, i.e.~we represent them as a Gaussian mixed state $\xi(\vec{\lambda},\pi)$ plus additional correlators corresponding to subsets of $S'$. Proceeding this way, we remove all correlators except for $I$ and end up expressing $\alpha C_S=\sum_{\vec{\lambda},\pi} w_{(\vec{\lambda},\pi)} \xi(\vec{\lambda},\pi)- \beta I$ for $\beta = 2^{-m}\sum_{\lambda,\pi} w_{\lambda,\pi}$. This procedure is probably far from optimal, i.e. gives rise to a large $c$. One can optimize this by applying the procedure to a density matrix $\rho$ from which we repeatedly take out Gaussian mixed states, i.e. $\rho \rightarrow \rho-w  \xi$ which remove the current highest-order correlator and the Gaussian state $\xi$ is chosen to be the one with minimum (but non-negative) weight ${\rm Tr}\, w \xi$. 
\end{proof}

Now that we have established the essential features of the set of convex-Gaussian states, the natural question that follows is how to determine whether a state has a decomposition in terms of Gaussian pure states. As for separability, there are two sides to this question. One is to find a Gaussian decomposition, a problem which is likely, as in the entanglement case, to be computationally hard in general. The reverse question is to find a criterion which establishes that a state is not convex-Gaussian. Note that the goal is to find a criterion which acts linearly on the pure Gaussian states, so that it can naturally be extended to convex mixtures thereof. The Hermitian operator $\Lambda=\sum_{i=1}^{2m} c_i \otimes c_i \in \cc C_{2m} \otimes \cc C_{2m}$, introduced in Ref.~\cite{bravyi:FLO}, will be useful in this context as it has been shown to lead to a necessary and sufficient criteria for a state to be Gaussian:

\begin{lemma}[Bravyi~\cite{bravyi:FLO}] An even state $\rho\in \cc C_{2m}$ is Gaussian iff $[\Lambda, \rho \otimes \rho]=0$.
\end{lemma}

\noindent The operator $\Lambda$ has the important property of being invariant under $U \otimes U$ where $U$ is any FLO transformation. The action of $\Lambda$ on $\rho \otimes \rho$ can  be appreciated by considering the trace norm $||.||_1$ of the positive operator $\Lambda (\rho\otimes\rho) \Lambda$:
 \begin{eqnarray}
|| \Lambda (\rho \otimes \rho) \Lambda||_1&=&\sum_{a,b=1}^{2m} \tr[(c_a\otimes c_a)(\rho\otimes \rho)(c_b\otimes c_b)]\nonumber\\
&=& 2m-\sum_{a \neq b} \left( {\rm Tr} \, \ii c_a c_b \rho\right)^2\nonumber\\
&=&2m- \sum_{a, b} (M_{ab})^2=2m-{\rm Tr} \,M^T M,
\label{eq:LrrL}
 \end{eqnarray}
where $M$ is the correlation matrix of the state $\rho$. For a pure Gaussian state $\rho=\sigma$, we have $M^T M=I$ implying that $\Lambda(\sigma \otimes \sigma)=0$. For a mixed Gaussian state or non-Gaussian state $M^T M < I$, see Lemma \ref{lem:M}, and hence $||\Lambda (\rho \otimes \rho) \Lambda ||_1 > 0$. These observations are collected in the following corollary.

\begin{corollary}
For an even state $\rho\in \cc C_{2m}$,  $\Lambda (\rho \otimes \rho)=0$ iff $\rho$ is a pure Gaussian state. 
\label{lem:iffgaus}
\end{corollary}

Corollary \ref{lem:iffgaus} says that the state of two copies of pure Gaussian states is contained in the null-space of the operator $\Lambda$. In Lemma~\ref{lem:gauss_sym} in the Appendix A we will prove that the null-space of the operator $\Lambda$ is contained in the symmetric subspace ${\bf Sym}^2(\Cx^{2^m})$ spanned by vectors $\ket{\psi}$ which are invariant under the SWAP operator $P$, i.e. $\ket{\psi}=P \ket{\psi}$. In particular, the null-space of $\Lambda$ is spanned by states $\ket{\psi,\psi}$ where $\psi$ is Gaussian, while the symmetric subspace is spanned by $\ket{\phi,\phi}$ for any $\phi$. We will call the null-space of $\Lambda$ the Gaussian-symmetric subspace. As $\Lambda$ is a sum of mutually commuting Hermitian operators $c_i \otimes c_i$ with $\mu_i=\pm 1$ eigenvalues, the projectors on the eigenstates are $P_{\vec{\mu}}=\frac{1}{2^{2m}}\Pi_{i=1}^{2m} ( I + \mu_i c_i \otimes c_i)$ with eigenvalue $\sum_{i=1}^{2m} \mu_i$ of $\Lambda$. Hence the Gaussian-symmetric subspace is spanned by $2m \choose m$ orthogonal eigenvectors $P_{\vec{\mu}}$ with the property that $\sum_i \mu_i=0$. Clearly, the dimension of the Gaussian-symmetric subspace ${2m \choose m}$ is smaller than the dimension of the symmetric subspace ${\bf Sym}^2(\Cx^{2^m})$ which is ${2^m+1 \choose 2}$ for any $m \geq 1$.

In \cite{DPS:prl, DPS:long} the authors established a series of tests for separability based on the existence of a symmetric extension of any separable state. The usefulness of these tests relies on the fact that they correspond to semi-definite programs which for small numbers of extensions can be implemented numerically. Furthermore, these tests are complete, in the sense that an entangled state does not have a symmetric extension to an arbitrary number of parties in the extension. Here we  formulate a similar series of extension tests which are all passed by convex-Gaussian states, but  non-convex-Gaussian states fail in some of them. Our criterion is based on the following simple observation:  Let a density matrix $\rho\in \cc C_{2m}$ be convex-Gaussian, i.e. we can write $\rho=\sum_i p_i \sigma_i$ with $\sigma_i$  pure Gaussian states, then there exists a symmetric extension $\rho_{ext} \in \cc C_{2m}^{\otimes n}$, namely $\rho_{\rm ext}= \sum_i p_i \sigma_i^{\otimes n}$, which is annihilated by $\Lambda$ acting on any pair of tensor-factors, and $\tr_{2,\ldots, n-1}\rho_{ext}=\rho$.
This immediately leads to the following feasibility semi-definite program:
\begin{program}
\label{prog:SExtSDP}
\begin{tabular}[t]{ll}
Input:&$\rho \in \cc C_{2m}$ and an integer $n\ge 2$\\
Body:&Is there a $\rho_{\rm ext} \in \cc C_{2m}^{\otimes n}$\\
&s.t. \begin{tabular}[t]{l}
$\tr_{2,\ldots, n}\rho_{ext}=\rho$\\
$\tr \,\rho_{\rm ext}=1$\\
$\Lambda^{k,l} \rho_{ext}=0, \forall k \neq l$\\
$\rho_{\rm ext}\geq 0$
\end{tabular}\\
Output:& \verb+yes+ (provide $\rho_{\rm ext}$) or \verb+no+
\end{tabular}
\end{program}
\noindent As there exists an isomorphism between $\mathcal{C}_{2m}^{\otimes n}$ and $\mathcal{C}_{2m n}$ (see the explicit map in~\cite{bravyi:FLO} for $\mathcal{C}_{2m} \otimes \mathcal{C}_{2m}$ and $\mathcal{C}_{4m}$), one could alternatively express program 1 as an extension of $\rho$ to a physical system with $2m n$ Majorana fermions (see~\cite{anna13}).

Note that in program~\ref{prog:SExtSDP} we do not need to enforce any permutation or Bose-symmetry on the extension \footnote{Adding the permutation-symmetry would of course leave the program in SDP form.} because of the following. By Lemma \ref{lem:gauss_sym} in Appendix A, the null-space of $\Lambda_{k,l}$ is spanned by vectors $\ket{\psi,\psi}_{k,l}$ where 
$\ket{\psi}$ is any pure Gaussian state. Thus the intersection of null-spaces of all $\Lambda_{k,l}$ is spanned by 
vectors $\ket{\psi^{\otimes n}}$ where $\psi$ is a pure Gaussian state, which are vectors in the symmetric subspace of $n$ parties. As $\rho_{ext}$ lies in this null-space, it will thus be Bose-symmetric \cite{CKMR:definetti}, i.e. $\pi \rho_{ext}=\rho_{ext}$ where $\pi$ is any permutation in $S_n$.

It is easy to see that any state $\rho$ has a Bose-symmetric extension, hence the stronger constraint imposed by $\Lambda$ is crucial. We say that a even state $\rho$ has a $n$-Gaussian-symmetric extension if Program~\ref{prog:SExtSDP} returns \verb+yes+ for inputs $\rho$ and $n$.  An explicit construction of the SDP for the case $n=2$ is given in Appendix B, and  some numerical results are discussed in the following section. Clearly all convex-Gaussian states have $n$-Gaussian-symmetric extensions for all $n$. Here we will prove that if a state has a $n$-Gaussian-symmetric extension for all $n$, then its distance to the set of convex-Gaussian states converges to zero.

\begin{theorem}
If an even state $\rho\in \cc C_{2m}$ has a $n$-Gaussian-symmetric extension for all $n$, then there exists a sequence of convex-Gaussian states $ \omega_1, \omega_2,\ldots \in  \mathcal{C}_{2m}$ such that $\lim_{n\rightarrow \infty} ||\rho - \omega_n||_1=0$.
\end{theorem}

\begin{proof}

To prove that we will invoke the quantum de Finetti theorem~\cite{CKMR:definetti}.
Let $\rho_{\rm ext}^n \geq 0$ be the extension of $\rho$ to $\cc C_{2m}^{\otimes n}$ such that $\Lambda^{k,l}\rho_{\rm ext}^n=0$ and ${\rm Tr}_{2,\ldots, n} \rho_{\rm ext}^n=\rho$. Using Lemma 4 in the appendix this shows that $\rho_{ext}^n$ is Bose-symmetric. Let $\rho_{1,2}^n$ be the reduced density matrix for the first two tensor factors, i.e. $\rho_{1,2}^n={\rm Tr}_{3,\ldots,n} \rho_{ext}^n$. Then Theorem II.8 in \cite{CKMR:definetti} tells us that there exists a probability distribution $\{\gamma_a(n)\}$ and states $\{\tau_a(n)\}\subset \mathcal{C}_{2m}$ such that
\begin{equation}
||\rho_{1,2}^n-\sum_a \gamma_a(n) \tau_a(n) \otimes \tau_a(n)||_1  \leq \epsilon_m(n),
\end{equation}
with $\epsilon_m(n)=\frac{4 \cdot 2^m}{n}$. Since the trace-norm can not increase by taking the partial trace, then 
\begin{equation}
\label{eq:GoToTau}
||\rho-\sum_a \gamma_a(n) \tau_a(n)||_1  \leq \epsilon_m(n).
\end{equation}
Therefore, in the limit of $n\rightarrow \infty$ we have that   $\sum_a \gamma_a(n) \tau_a(n)$  converges  to $\rho$.

On the other hand, we  have
\beqa
||\, \Lambda \sum_a \gamma_a(n) \tau_a(n) \otimes \tau_a(n) \Lambda\,||_1 & = & ||\,\Lambda (\rho_{1,2}^n-\sum_a \gamma_a(n) \tau_a(n) \otimes \tau_a(n)) \Lambda\,||_1 \nonumber\\
& \leq & ||\Lambda||_1^2 \epsilon_m(n) \equiv \tilde{\epsilon}_m(n) ,
\label{eq:GoToGauss}
\eeqa
with $||\Lambda||_1^2=4(m+1)^2 {2m \choose m+1}^2$. For fixed $m$, in the limit $n \rightarrow \infty$ the upper-bound $\tilde{\epsilon}_m(n) \rightarrow 0$. Corollary 1 then implies  that for each $a$ either $\gamma_a(n)\rightarrow 0$ or $\tau_a(n)$ converges to a pure Gaussian state.  That means that there exists a series of pure Gaussian states $\upsilon_a(1), \upsilon_a(2),\ldots \in \mathcal{C}_{2m}$ such that  $\lim_{n\rightarrow \infty} ||\tau_a(n) -\upsilon_a(n)||_1=0$.

Defining the convex-Gaussian states $\omega_n=\sum_a \gamma_a(n) \upsilon_a(n)$, from Eq.(\ref{eq:GoToTau}) we have:
\beqa
\epsilon_m(n) & \geq & ||\rho-\omega_n +\omega_n-\sum_a \gamma_a(n) \tau_a(n)||_1\\
			 & \geq & \Big|\; ||\rho-\omega_n||_1 - ||\omega_n-\sum_a \gamma_a(n) \tau_a(n)||_1\; \Big|.
\eeqa
Where we used the reversed triangle inequality. Therefore:
\beqa
||\rho-\omega_n||_1 & \leq & \epsilon_m(n) +||\omega_n-\sum_a \gamma_a(n) \tau_a(n)||_1\\
			        & \leq & \epsilon_m(n) +\sum_a \gamma_a(n) ||\upsilon_a(n)- \tau_a(n)||_1.
\eeqa
The claimed result then follows by noting that both $\epsilon_m(n)$ and $||\upsilon_a(n)- \tau_a(n)||_1$ go to zero for $n\rightarrow \infty$.
\end{proof}

\begin{remark} A quantitative version of this result will be presented somewhere else~\cite{anna13}.
\end{remark}

Note that convex-Gaussian states always have a symmetric extension which is separable in relation to all tensor-factors. This property can be used to provide an alternative formulation of the necessary and sufficient criterion for a state to be convex-Gaussian:

\begin{proposition} \label{theo:GaussEnt} A state $\rho \in \cc C_{2m}$ is convex-Gaussian iff there exists an extension $\rho_{\rm ext}\in \cc C_{2m}\otimes \cc C_{2m}$ that is a feasible solution to Program~\ref{prog:SExtSDP} ---input $\rho$ and $n=2$--- with the additional constraint that $\rho_{\rm ext}$ is separable between the tensor-factors.
\end{proposition}
\begin{proof}
As before one direction is immediate. For the other direction, assume that there exists a state $\rho_{\rm ext}\in \cc C_{2m}\otimes \cc C_{2m}$ which is separable, i.e. it is of the form $\rho_{\rm ext}= \sum_i p_i \tau_i \otimes \varsigma_i$, with $\{\tau_i\}$ and $\{\varsigma_i\}$  sets of density matrices in $\cc C_{2m}$ and $\{p_i\}$ a probability distribution. From the constraint $\Lambda \rho_{\rm ext}=0$, it follows that 
\begin{eqnarray}
0&=&{\rm Tr} \Lambda \left( \sum_i p_i \tau_i \otimes \varsigma_i\right) \Lambda \nonumber \\
&=&2m-\sum_i p_i \sum_{a \neq b}({\rm Tr}\; \ii c_a c_b \tau_i)({\rm Tr} \; \ii c_a c_b \varsigma_i)\nonumber\\
&=&2m-\sum_i p_i {\rm Tr} M^T(\tau_i)  M(\varsigma_i).
\end{eqnarray}
Where the notation $M(\rho)$ represents the correlation matrix of the state $\rho$.

Using the triangle and Cauchy-Schwartz inequalities, we have 
\begin{eqnarray}
2m &=& |\sum_i p_i {\rm Tr} M^T(\tau_i) M(\varsigma_i)|\nonumber\\
&\leq& \sum_i p_i \sqrt{{\rm Tr} M^T(\tau_i) M(\tau_i)}\sqrt{ {\rm Tr} M^T(\varsigma_i) M(\varsigma_i)}
\label{eq:hold}.
\end{eqnarray}
From Lemma \ref{lem:M} it follows that ${\rm Tr} M^T M=2\sum_{i=1}^m |\lambda_i|^2  \leq 2m$ with equality iff $M$ is the correlation matrix of a pure Gaussian state. In order for equality to hold in Eq.~(\ref{eq:hold}), each $\tau_i$ (and $\varsigma_i$) must thus be a pure Gaussian state and therefore $\rho=\sum_i p_i \tau_i$ is convex-Gaussian.
\end{proof}

Including this extra separability constraint in  Program~\ref{prog:SExtSDP} breaks its SDP structure -- checking separability can in itself be cast as an SDP~\cite{DPS:long}, but this would give a nested SDP structure to the program deciding if a state is convex-Gaussian. A direct semi-definite relaxation of this criteria is obtained by employing the partial-transposition test for separability~\cite{PPT:peres,PPT:horodecki}, leading to the following SDP:
\begin{program}
\label{prog:PPTSDP}
\begin{tabular}[t]{ll}
Input:&$\rho \in \cc C_{2m}$\\
Body:&Is there $\rho_{\rm ext} \in \cc C_{2m}\otimes\cc C_{2m} $\\
&s.t. \begin{tabular}[t]{l}
$\tr_{2}\rho_{ext}=\rho$\\
$\tr \rho_{\rm ext}=1$\\
$\Lambda \rho_{ext}=0 $\\
$\rho_{\rm ext}\geq 0$\\
$\rho_{\rm ext}^{T_1}\geq 0$
\end{tabular}\\
Output:& \verb+yes+ (provide $\rho_{\rm ext}$) or \verb+no+
\end{tabular}
\end{program}
\noindent Here $T_1$ is the transposition map on the first tensor-factor. As for the previous program, a negative answer means that the state is not convex-Gaussian. For the converse,  it is likely that the PPT criterion together with the constraint that the extension lives in the Gaussian-symmetric subspace is not sufficient to enforce separability (e.g., there exists symmetric bound-entangled states~\cite{Geza09}), hence a positive answer is non-decisive, but we leave this as an open question.


\section{Applications: depolarized $\ket{a_8}$}
\label{sec:appl}

As mentioned in the introduction, one way to turn $\nu=5/2$ topological quantum computation into an universal quantum computation scheme, is to assume that one has access to the auxiliary states $\ket{a_4}\in \cc C_4$ and $\ket{a_8}\in \cc C_8$, as shown in~\cite{bravyi:uni_maj}. In the present section we want to assess the amount of noise that can be added to these extra resources before they turn convex-Gaussian. From Proposition~\ref{prop:123Gauss} we know that any pure state in $\cc C_4$ is Gaussian, and therefore any noisy version of $\ket{a_4}$ is convex-Gaussian. We thus concentrate on the state $\ket{a_8}$ when undergoing depolarization. We consider the state
\beq
\label{eq:depoA8}
\rho_{a_8}(p) = (1-p)\proj{a_8} + p I/16,
\eeq
with $0 \le p \le 1$ and 
\begin{equation}
\proj{a_8}=\frac{1}{16}(I+S_1)(I+S_2)(I+S_3)(I+Q),
\label{def:a8}
\end{equation}
with $S_1=-c_1 c_2 c_5 c_6$, $S_2=-c_2 c_3 c_6 c_7$, $S_3=-c_1 c_2 c_3 c_4$ and $Q=c_1 c_2 c_3 c_4 c_5 c_6 c_7 c_8$. Ref.~\cite{bravyi:uni_maj} has shown that $\rho_{a_8}(p)$ can be distilled for $p < 40\%$ (corresponding to the parameter $\epsilon_8$ defined in Ref.~\cite{bravyi:uni_maj} equal to $0.38$). We would like to determine the noise threshold $p_*$ such that  $\rho_{a_8}(p\geq p_*)$ is convex-Gaussian. Note that $\rho_{a_8}$ is Gaussian only for $p=1$. The results in this Section are summarized in Fig.\ref{fig:regions}.

\begin{figure}[h!tb]
\centerline{
\mbox{
\includegraphics[width=0.9\linewidth]{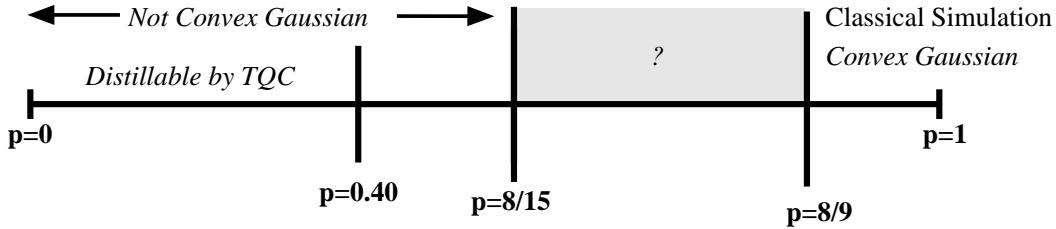}}}
\caption{Properties of the state $\rho_{a_8}(p)=(1-p)\proj{a_8}+p I/16$.  Our results show that $p_*$, the noise threshold above which $\rho_{a_8}(p)$ is convex-Gaussian is in the interval $[\frac{8}{15}, \frac{8}{9}]$, depicted in grey. In this interval we do not yet know whether $\rho_{a_8}$ is convex-Gaussian. For $p < 0.40$, Ref.~\cite{bravyi:uni_maj} shows that $\rho_{a_8}$ can be distilled by noiseless braiding operations to enable topological quantum computation (TQC) with Clifford group elements (for quantum universality one also needs a sufficiently low-noise $\ket{a_4}$ ancilla). Our results from the SDP on the existence of a Gaussian-symmetric extension of $\rho_{a_8}(p)$ for $p \geq 0.25$ do not provide additional information beyond these analytical results.}
\label{fig:regions}
\end{figure}

In order to give a lower-bound on $p_*$ we implemented the SDP Program~\ref{prog:SExtSDP} for fixed $n=2$ (see Appendix B) with the aid of the {\sc cvxopt} library for Sage~\cite{Sage}. The numerical results indicate that $p_*\ge 0.25$, as it is always possible to find a Gaussian-symmetric extension for $\rho_{a_8}(p)$ with $n=2$ for noise rates above this value. Our implementation of the feasibility SDP takes approximately one day for each value of $p$ in a computer with $2.3GHz$ processor, and takes about $30GB$ of RAM memory. The implementation of the SDP with an additional PPT condition on the extension is likely to provide stronger results, but it is  more computationally intensive. Given these limitations, we provide an alternative analysis of the properties of $\rho_{a_8}(p)$ using the notion of witnesses: Hermitian operators which have expectation values in a restricted range (e.~g.~ non-negative for all Gaussian states) whereas expectation values on non-Gaussian states can extend beyond this range (e.g.~be negative). The existence of such witness operators is guaranteed by the fact that they represent separating hyperplanes separating the convex-Gaussian states from all even Hermitian operators in ${\cal C}_{2m}$ (in addition, similar as in~\cite{DPS:long}, witnesses could be constructed based on a negative output of feasibility SDP).

Clearly, such Hermitian witness operators, say $W$, should include terms which are quartic or higher-weight correlators, as the expectation value of quadratic Hamiltonians only depends on the correlation matrix of a state. For $\rho_{a_8}(p)$ the obvious choice for such a witness operator is $W=\ket{a_8}\bra{a_8}$, for which ${\rm Tr} W \rho_{a_8}(p)=1-\frac{15p }{16}$. Let us understand how we can bound  $\min_{\psi_g} {\rm Tr} W \ket{\psi_g}\bra{\psi_g}=|\bra{\psi_g} a_8 \rangle|^2$ where $\psi_g$ is any Gaussian pure state. We will use the fact that the correlation matrix of $\ket{a_8}$ is the null matrix, i.e. for all Majorana operators $c_i, c_j$, ${\rm Tr}\, c_i c_j \ket{a_8}\bra{a_8}=0$, which follows directly from Eq.~(\ref{def:a8}). We will now prove that 
$|\bra{\psi_g} a_8 \rangle|^2 \leq \frac{1}{2}$ which shows that for $p < \frac{8}{15} \approx 0.53$, the state $\rho_{a_8}(p)$ is not convex-Gaussian. The following proposition proves our claimed bound and can be intuitively understood by interpreting the state $\ket{a_8}$ as a maximally entangled state while Gaussian states are product states:

\begin{proposition}
For all Gaussian pure states $\ket{\psi_g}$, we have $|\bra{\psi_g} a_8 \rangle|^2 \leq \frac{1}{2}$.
\label{prop:maxE}
\end{proposition}

\begin{proof}
Let $c_1$ and $c_2$, $c_1 c_2=-c_2 c_1$, be two arbitrary Majorana fermion operators, from a complete set $\{c_{2k-1},c_{2k}\}_{k=1}^4$ (note the $SO(8)$ freedom in choosing these). Let $\ket{x}$ denote an eigenvector of the mutually commuting operators $\ii c_1 c_2$, $\ii c_{2k-1} c_{2k}$ for $k=2,\ldots 4$ where the bits (0 or 1) of $x$ label the eigenvalues (resp. $+1$ and $-1$) of these operators. The Gaussian states $\ket{x}$ form a basis for the $4$-fermion (or 4-qubit) space, hence we can write $\ket{a_8}=\sum_{x \in \{0,1\}^4} \alpha_x \ket{x}$ with $\sum_{x} |\alpha_x|^2=1$. Then 
\begin{equation}
{\rm Tr}\, \ii c_1 c_2 \proj{a_8}=\sum_{y\in \{0,1\}^3} |\alpha_{0y}|^2-\sum_{y \in \{0,1\}^3} |\alpha_{1y}|^2=0.
\end{equation}
This, together with normalization, implies that $\sum_y |\alpha_{0y}|^2=\sum_y |\alpha_{1y}|^2=\frac{1}{2}$, hence for all $y$, $|\alpha_{0y}|^2$ and $|\alpha_{1y}|^2$ are less than or equal to $1/2$.
\end{proof}

A tighter bound than the one in Proposition \ref{prop:maxE} cannot be excluded if one uses further properties of the state $\ket{a_8}$ -- such tighter bound would lead to a reduction of the grey area in Fig.\ref{fig:regions}.

To give an upper-bound on $p_*$ we construct explicit Gaussian decompositions of $\rho_{a_8}(p)$ for large values of $p$.  To do this we consider a subset of the convex-Gaussian states that share some key properties with $\rho_{a_8}(p)$, namely: (i)  zero correlation matrix, and (ii) only a small fraction of all possible correlators has non-zero coefficients.

To exploit these symmetries we define three types of states:
\beq
\rho_i({\lambda}) = \frac{1}{2} \left(\rho_{M_i({\lambda})} + \rho_{M_i(- {\lambda})}\right),
\eeq
where $\rho_{M_i ({\lambda})}$, for $i=1,2,3$, is the Gaussian state generated by the correlation matrix $M_i({\lambda})$, with ${\lambda}\in [-1,1]^4$, and $M_i({\lambda})$ assumes one of the three following forms:
\beqa
M_1({\lambda}) = \bigoplus_{i=1}^{4}\left( \begin{array}{cc}
							0&\lambda_i\\
							-\lambda_i & 0
							\end{array}\right),\nonumber\\
							M_2({\lambda}) = \bigoplus_{i=1}^{2}\left( \begin{array}{cccc}
														0&0&\lambda_i &0\\
														0&0&0&\lambda_{i+1}\\
														-\lambda_i &0&0&0\\
														0&-\lambda_{i+1}&0&0
														\end{array}\right), \nonumber \\
														M_3({\lambda}) = \bigoplus_{i=1}^{2} \left(\begin{array}{cccc}
																				0&0&0&\lambda_i \\
																				0&0&\lambda_{i+1}&0\\
																				0&-\lambda_{i+1}&0&0\\
																				-\lambda_{i}&0&0&0
																				\end{array}\right).\nonumber
\eeqa
By construction, these states are convex-Gaussian, have null correlation matrix, and any convex combination of them has an expansion in terms of Majorana operators with non-vanishing coefficients only on the same correlators as $\rho_{a_8}(p)$. The aim is to find for which $p$'s it is possible to decompose  $\rho_{a_8}(p)$ as a convex sum of these type of convex-Gaussian states.

In order to describe the convex hull of these states, first note that $M_2$ and $M_3$ are related to $M_1$ by   permutations (different choices of pairings):
\beqa
M_2(\lambda)= \left(P_{12}\oplus P_{12}\right) M_1(\lambda) \left(P_{12}\oplus P_{12}\right)^T, \nonumber \\ M_3(\lambda)= \left(P_{13}\oplus P_{13}\right) M_1(\lambda) \left(P_{13}\oplus P_{13}\right)^T,
\nonumber
\eeqa
where
\beq
P_{12}=\left(\begin{array}{cccc}
		1&0&0&0\\
		0&0&1&0\\
		0&1&0&0\\
		0&0&0&1
		\end{array}\right),
						P_{13}=\left(\begin{array}{cccc}
						1&0&0&0\\
						0&0&0&1\\
						0&1&0&0\\
						0&0&1&0
						\end{array}\right).
\eeq
Given that the transformation connecting the correlation matrices is formed by the direct sum of identical permutations, this transformation is in $SO(8)$ -- one always gets the signature of the permutation squared, and therefore the corresponding determinant is always 1. 
Rotations in $SO$ of the correlation matrix  induce unitary transformations on the level of the states~\cite{bravyi:FLO}: $P_{12} \mapsto U_2$, $P_{13} \mapsto U_3$. Bearing  in mind  this connection among the three classes of states above, we now determine the extreme points of the set of type 1 states.

\begin{lemma}[Extreme points for type 1 states] Let $\overline{\cc S}_1 = \conv\{\rho_1(\lambda) | \lambda\in [-1,1]^4\}$. Then the extreme points of $\overline{\cc S}_1$ are given by
$$\phi_1(x) \doteq \frac{1}{2} \left( \proj{x} + \proj{\neg x}\right),$$
with $x \in \{0,1\}^4$ and $\neg x$ is the bitwise negation of $x$.
\end{lemma}
\begin{proof}
It follows immediately from the expansion of the state $\rho_1(\lambda)$:
\beqa
\rho_1(\lambda) &=&\frac{1}{2^5}\sum_{x\in \{0,1\}^4} \left[   \prod_{k=1}^4 \left(1 - (-1)^{x_k}\lambda_k\right) +  \prod_{k=1}^4 \left(1 - (-1)^{x_k}(-\lambda_k)\right) \right] \proj{x}\nonumber\\
&=&\frac{1}{2^5} \sum_{x\in \{0,1\}^4} \prod_{k=1}^4 \left(1 - (-1)^{x_k}\lambda_k\right)\left( \proj{x} +\proj{\neg x}\right)\\
&=&\frac{1}{2^4} \sum_{x\in \{0,1\}^4} \prod_{k=1}^4 \left(1 - (-1)^{x_k}\lambda_k\right)\phi_1(x).\nonumber
\eeqa
\end{proof}

Defining in a similar fashion $\overline{\cc S}_2$, $\overline{\cc S}_3$, $\phi_2(x)$ and $\phi_3(x)$, we can then define the convex hull of any convex combination of the three types of states as $\overline{\cc S} = \conv\{\overline{\cc S}_1,\overline{\cc S}_2,\overline{\cc S}_3\}$. More explicitly: 
\beq
\overline{\cc S}= \left\{\sum_{i=1}^3 \sum_{x=0}^7 \gamma_i(x) \phi_i(x) \Big | \; \gamma_i(x)\ge 0,  \sum_{i=1}^3 \sum_{x=0}^7  \gamma_i(x)=1\right \},\nonumber
\eeq
where the summation in $x$ is to be carried out over its binary expansion in three bits.

The problem now reduces to determine for which values of $p\in[0,1]$ the linear system $\rho_{a_8}(p) = \sum_{i=1}^3 \sum_{x=0}^7 \gamma_i(x) \phi_i(x)$, with $\{\gamma_i\}$ a probability distribution, has a solution. Simple inspection shows that such Gaussian decomposition is always possible whenever $p\ge 8/9 \approx 0.89$, which is then an upper-bound on $p_*$.

\section{Discussion:  Classical simulation of fermionic quantum computation}
\label{sec:discussion}

Contrary to bosonic linear optics, quantum computations based on fermionic linear optics (FLO) can be efficiently simulated by a classical computer even when augmented by measurements of  fermionic number operators~\cite{BarbaraFLO,JozsaMiyake,bravyi:FLO}. The simulation relies on the following facts:
\begin{enumerate}
\item efficient encoding of Gaussian states -- a Gaussian state of $2m$ Majorana fermions is fully described by its correlation matrix $M$, with $O(m^2)$ elements.
\item FLO transformations map Gaussian states onto Gaussian states, with efficient update rule -- every FLO transformation $V\in \cc C_{2m}$ induces a rotation $R\in SO(2m)$ on the Majorana operators space, $V c_i V^\dagger = \sum_{j=1}^{2m} R_{ij} c_j$. This, in turn, induces the map $M\mapsto RMR^T$ on the $2m\times 2m$ correlation matrix, and this update can be evaluated in $O(m^3)$ steps.
\item efficient read-out of measurement probability distributions -- via Wick's theorem~(\ref{wick}) the probability of measuring the population in fermionic modes can be done efficiently, as the Pfaffian of a $2m\times 2m$ matrix can be evaluated in $O(m^3)$. Furthermore, number operator measurements project Gaussian states onto Gaussian states.
\end{enumerate}

With these three ingredients, the evolution of any initial Gaussian state can be efficiently followed by a classical computer at any time in the computation, and the output distribution can be exactly evaluated. In references~\cite{bravyi:FLO,BravyiKoenig} the allowed transformations in (ii)   were extended to noisy \emph{Gaussian maps}. These are completely-positive channels, not necessarily FLO unitaries, that map Gaussian states onto Gaussian states, and as such still admit efficient classical simulation.

Our results imply that if $p\ge 8/9$, $\rho_{a_8}(p)$ is convex-Gaussian and this allows for the classical simulation of the computation as follows. Given a noise strength $p\ge 8/9$, one finds the decomposition of $\rho_{a_8}(p)$ in $\overline{\cc S}$, by solving the linear system $\rho_{a_8}(p)=\sum_{i=1}^3\sum_{x=0}^7 \gamma_{i}(x) \phi_i(x)$ -- this has to be done only once for a fixed $p$. As each $\phi_i(x)$ is in itself a convex sum of two pure Gaussian states, this leads to a decomposition of $\rho_{a_8}(p)$ in terms of at most 48 pure Gaussian states. Let   $p_i$ be the probability associated with the $i$-th pure Gaussian state in such decomposition. Then, whenever a state $\rho_{a_8}(p)$ is requested in the computation, one samples from $\{p_i\}$ and chooses the corresponding pure Gaussian state. From this point onwards, the simulation follows the scheme summarized in  (i), (ii) and (iii) above.

\medskip

The results in Fig.~\ref{fig:regions} leave some gaps in our understanding, in particular about the precise value of $p_*$. In order to improve the lower bounds on $p_*$ one could consider distillation processes beyond the restricted set of braiding operations, but that still map Gaussian states onto Gaussian states, leaving thus the set of convex-Gaussian states unchanged. The distillation protocol of depolarized $GHZ_4$ states, which are isomorphic to $\rho_{a_8}(p)$, proposed by D\"ur and Cirac at first sight suggests that our classical simulation threshold is tight, as it can distill a perfect $GHZ_4$ provided that $p<8/9$. Nevertheless, their protocol uses non-FLO operations, as it assumes that one can locally create the $GHZ_4$ and uses distilled 2-qubits maximally entangled pairs for distribution. The protocol devised by Murao \emph{et al.} in~\cite{murao}, directly distills $GHZ_4$ from depolarized copies of it for $p<0.7328$. Despite the fact that this direct approach seems more related to our question, it employs gates that are not immediately translated into braiding operations or even to FLO transformations. Although the choice to distill via braiding operations is  physically motivated -- due to their topological protection--, a different distillation threshold from FLO operations would suggest that either Bravyi's distillation protocol can be improved to higher depolarizing noise rates, or that there exists a `transition zone' in which the computation by braiding of Majorana fermions is neither quantumly-universal nor can be efficiently simulated classically.


\section*{Acknowledgements}
\addcontentsline{toc}{section}{Acknowledgements}
We would like to acknowledge fruitful discussions with Akimasa Miyake, Ronald de Wolf, Sergey Bravyi, and Volkher Scholz. We also would like to thank Maarten Dijkema for  IT support. F. d-M. is supported by the Vidi grant 639.072.803 from the Netherlands Organization for Scientific Research (NWO).


\setcounter{section}{1}
\appendix

\section{Symmetric versus Gaussian-symmetric Subspace}
\label{proof:twirl}

Inspired by~\cite{twirl:werner}, we define a `FLO twirl' as the map ${\cal S}(\rho)=\frac{1}{{\rm Vol}(U)}\int_{FLO} dU U \otimes U \rho U^{\dagger} \otimes U^{\dagger}$ where $\int_{FLO} dU$ is defined by taking the Haar measure over the real orthogonal matrices $R$ induced by $U$. The map ${\cal S}(\rho)$ is normalized such that it is trace-preserving.

\begin{lemma}
The projector onto the null-space of $\Lambda=\sum_i c_i \otimes c_i$ is $\Pi_{\Lambda=0}={2m \choose m} {\cal S}(\ket{0,0}\bra{0,0})$ where $\ket{0}$ is a (Gaussian) vacuum state with respect to some set of annihilation operators $a_i$. Thus the states $\ket{\psi,\psi}$ where $\psi$ is a pure fermionic Gaussian state span the null-space of $\Lambda$ which implies that the null-space of $\Lambda$ is a subspace of the symmetric subspace.
\label{lem:gauss_sym}
\end{lemma}

\begin{proof}
In order to prove that $\Pi_{\Lambda=0}={2m \choose m} {\cal S}(\ket{0,0}\bra{0,0})$, we note that both the l.h.s. and the r.h.s. are $U \otimes U$-invariant where $U$ is any FLO transformation. Thus instead of considering whether ${\rm Tr} X \Pi_{\Lambda=0}={\rm Tr}(X {2m \choose m} {\cal S}(\ket{0,0}\bra{0,0}))$ for any $X$, we can just consider the trace with respect to invariant objects ${\cal S}(X)$.  It can be shown that $\Lambda^i$ for $i=0,1,2,\ldots,2m$ (and linear combinations thereof) are the only invariants under the group $U \otimes U$ where $U$ is FLO transformation \cite{scholz:priv}. Clearly, ${\rm Tr} \Lambda^i {2m \choose m} {\cal S}(\ket{0,0}\bra{0,0})=0={\rm Tr} \Lambda^i \Pi_{\Lambda=0}$ for all $i \neq 0$ while ${\rm Tr} \Pi_{\Lambda=0}={2m \choose m}$ fixes the overall prefactor. Having established the form of the projector, it follows directly that the states $\ket{\psi,\psi}$ for any Gaussian $\psi$ span the null-space (Assume this is false and hence a state in the null-space $\ket{\chi}=\ket{\chi_{in}}+\ket{\chi_{out}}$ where $\ket{\chi_{in}}$ is in the span of $\ket{\psi,\psi}$ while $\ket{\chi_{out}}$ is w.l.o.g. orthogonal to any $\ket{\psi,\psi}$. We have $\Pi_{\Lambda=0}\ket{\chi}=\ket{\chi}$ while ${2m \choose m} {\cal S}(\ket{0,0}\bra{0,0}) \ket{\chi}=\ket{\chi_{in}}$ arriving at a contradiction.) As $\ket{\psi,\psi}$ for Gaussian pure states $\psi$ span the null-space, and $P \ket{\psi,\psi}=\ket{\psi,\psi}$ with $P$ the SWAP operator, the null-space is a subspace of the symmetric subspace.
\end{proof}

\section{SDP to detect non-convex-Gaussian states}
\label{app:SDP}

Recall the series of tests constructed to detect that if a state $\rho \in \cc C_{2m}$ is not convex-Gaussian:
\begin{quote}
\begin{tabular}{ll}
Input:&$\rho \in \cc C_{2m}$ and an integer $n\ge 2$\\
Body:&Is there $\rho_{\rm ext} \in \cc C_{2m}^{\otimes n}$\\
&s.t. \begin{tabular}[t]{l}
$\pi \rho_{\rm ext} \pi^\dagger = \rho_{\rm ext}$, $\forall \pi \in S_n$\\
$\tr_{2,\ldots, n}\rho_{ext}=\rho$\\
$\tr \rho_{\rm ext}=1$\\
$\Lambda^{k,l} \rho_{ext}=0, \forall k \neq l$\\
$\rho_{\rm ext}\geq 0$
\end{tabular}\\
Output:& \verb+yes+ (provide $\rho_{\rm ext}$) or \verb+no+
\end{tabular}
\end{quote}
Note that we have explicitly kept the permutation-symmetry constraint in this program.

In this Appendix we convert this problem  into a feasibility semi-definite programming problem. To do that we must show that the above problem can be cast as the standard form of SDP's, namely:
\begin{quote}
\begin{tabular}{ll}
minimize& $\bb c^T \bb x$\\
subject to& $F_0 + \sum_i x_i F_i \geq 0$\\
&$A \bb x = \bb b$
\end{tabular}
\end{quote}
where $\bb x \in \Rl^d$ is the unknown vector to be optimized over, $\bb c$ is a  given vector in $\Rl^d$, and  $\{F_i\}_{i=0,\ldots, d}$ are given symmetric matrices. For the equality constraint, $A \in \Rl^{p\times d}$ is a given matrix, with ${\rm rank}(A)=p$, and $\bb b \in \Rl^p$. The restriction on the rank of $A$ demands that the linear system has at least one solution, and all the rows are linearly independent.

The first step is to associate an Hermitian operator $M_j$ with each term $\ii^k c_{a_1} c_{a_2}\ldots c_{a_{2k}}$ in the expression~(\ref{anyX}). The set $\{M_j\}_{0\le j \le 2^{2m -1}-1}$ spans  all the even Hermitian operators of $\cc C_{2m}$. Then any state $\rho\in \cc C_{2m}$ can be written as:
\beq
\rho = \sum_{k=0}^{2^{2m -1}-1} \alpha_k M_k,
\eeq
and we choose $M_0=I$.

For definiteness, in what follows we restrict to the case $n=2$. Extensions to larger $n$'s  follow the same structure and can be immediately constructed. That said, a general extension on $\cc C_{2m}\otimes \cc C_{2m}$ can be written as:
\beq
\rho_{\rm ext} = \sum_{i,j=0}^{2^{2m -1}-1} \beta_{ij}M_i\otimes M_j,
\eeq
with the coefficients $\beta_{ij}\in \Rl$ to be fixed by the constraints of our problem. The symmetry constraint can be taken directly into the parametrization of $\rho_{\rm ext}$ by imposing:
\beq
\rho_{\rm ext} = \sum_{i=0}^{2^{2m -1}-1} \beta_{ii}M_i\otimes M_i + \sum_{0\le i<j\le 2^{2m -1}-1} \beta_{ij}(M_i\otimes M_j + M_j\otimes M_i).
\eeq
At this point, the initially $2^{4m -2}$ coefficients of $\rho_{\rm ext}$ are reduced to $2^{2m-2}(1+2^{2m-1})$.

The imposition that $\tr_1(\rho_{\rm ext}) = \rho$ further demands:
\beq
2^m \sum_{j=0}^{2^{2m-1}-1} \beta_{0,j} M_j = \sum_{j=0}^{2^{2m-1}-1} \alpha_j M_j.
\eeq
Therefore, $\beta_{0,j}=\alpha_j/2^{m}$. In this way, further $2^{2m-1}$ coefficients are determined, and thus remain $2^{2m -3}(2^{2m}-2)$ free parameters. Note that the normalization of $\rho_{\rm ext}$ is already guaranteed as long as $\tr(\rho)=\alpha_0 =1$. 

The question whether there exists an assignment of the remaining free coefficients respecting the last two constraints  can be immediately  posed as a feasibility SDP:
\begin{quote}
\begin{tabular}{ll}
minimize & $0$\\
subject to&$\rho_{\rm ext}\geq 0$\\
&$\Lambda \rho_{\rm ext}=0$.
\end{tabular}
\end{quote}

Forgetting for the moment that the correlators are in general complex, a direct correspondence with the SDP standard form can be done as follows:
\beqa
F_0 \rightarrow \beta_{0,0} M_0\otimes M_0 + \sum_{j=1}^{2^{2m-1}-1}\beta_{0,j}(M_0\otimes M_j + M_j\otimes M_0)\\
F_i \rightarrow \left\{\begin{array}{l}
						M_j\otimes M_j, \;\; 1\le j \le 2^{2m-1}-1\\
						M_j\otimes M_k + M_k\otimes M_j, \;\; 1\le j< k\le 2^{2m-1}-1
						\end{array}\right.,
\eeqa 
with $\bb x=(\{\beta_{j,j}\}_{1\le j \le 2^{2m-1}-1}, \{\beta_{j,k}\}_{1\le j< k\le 2^{2m-1}-1})^T$.

The equality constraint then reads:
\beq 
\sum_{i=1}^{2^{2m -3}(2^{2m}-2)} x_i \Lambda F_i = -\Lambda F_0.
\eeq
This can be cast in the form $A \bb x = \bb b$, by writing $A=([\Lambda F_i]_{1\le i \le 2^{2m -3}(2^{2m}-2)})$, and $\bb b = -[\Lambda F_0]$. The notation $[G]$ means a column matrix constructed by stacking each column of $G$ below each other.

To finish the construction we must take into account that we are possibly dealing with complex values. To do that we employ the following well-known result.
\begin{lemma}
\label{prop:realMap} Let $Z=Z^\dagger\in \Cx^{d\times d}$. Then $Z\geq  0$ iff $\cc T(Z) =\left(\begin{array}{cc}
Re(Z)&-Im(Z)\\
Im(Z)&Re(Z)
\end{array}\right)$.
\end{lemma}
\begin{proof}
$(\Rightarrow)$ Let $\bb z = \bb x + \ii\bb y$ be an eigenvector of $Z$ with eigenvalue $\lambda$. Then  it is easy to see that $(\bb x,\bb y)^T$ is an eigenvector of $\cc T(Z)$ with eigenvalue $\lambda$.

$(\Leftarrow)$ By the same token, let $(\bb x,\bb y)$ be an eigenvector of $\cc T(Z)$ with eigenvalue $\lambda$. Then $\bb z= \bb x + \ii \bb y$ is an eigenvector of $Z$ with eigenvalue $\lambda$.
\end{proof}
The final dimension of each $F_i$ is then $2^{2m+1}\times 2^{2m+1}$, and there are $2^{2m -3}(2^{2m}-2)$ of such matrices.
For the equality constraint we write 
\beq
\left(\begin{array}{c}
Re(A)\\
Im(A)
\end{array}
\right) \bb x = \left(\begin{array}{c}
					Re(\bb b)\\
					Im(\bb b)
					\end{array}\right).
\eeq
The final matrix $(Re(A),Im(A))^T$ has dimension $2^{4m+1} \times 2^{2m -3}(2^{2m}-2)$. Most likely the rank of this matrix is not equal to the number of rows. Therefore, before  plugging it into the SDP one must evaluate its reduced-row form (echelon form), and remove the all-zero lines.

If one wishes to include the PPT test, as in Program~\ref{prog:PPTSDP}, this can be done by rewriting the two positivity constraints, $\rho_{\rm ext}\geq 0$ and $\rho_{\rm ext}^{T_1}\geq 0$, as a single one, $\rho_{\rm ext}\oplus \rho_{\rm ext}^{T_1}\geq 0$. One thus need to map each complex $F_i$ in the positivity constraint to $F_i^\prime = F_i\oplus F_i^{T_1}$. The final dimension of each real $F_i^\prime$ is then  $2^{2m+2}\times 2^{2m+2}$. The equality constraint remains the same.

\section*{References}
\addcontentsline{toc}{section}{References}
\bibliographystyle{unsrt}
\bibliography{noisyFLO}

\end{document}